\documentstyle{article}

\def\noheaderplainsetup{%

      \topmargin=0pt \headheight=0pt \headsep=0pt  
      \oddsidemargin=18pt \evensidemargin=18pt       
      \textheight=8.9truein \textwidth=5.9truein}  

\noheaderplainsetup

\begin{document}


\newcommand{\col}[1]{\mbox{$#1$:}}
\newcommand{\xx}{\wp} 
\newcommand{\pintimpl}{\mbox{\hspace{2pt}\raisebox{0.033cm}{\tiny $>$}\hspace{-0.18cm} \raisebox{-0.043cm}{\large --}\hspace{2pt}}}
\newcommand{\chess}{\mbox{\em Chess}}
\newcommand{\pintimpli}{\mbox{\hspace{2pt}\raisebox{0.033cm}{\tiny $>$}\hspace{-0.18cm} \raisebox{-0.015cm}{\large -}\hspace{2pt}}}

\newcommand{\mla}{\mbox{{\Large $\wedge$}}}
\newcommand{\mle}{\mbox{{\Large $\vee$}}}
\newcommand{\psti}{\mbox{\raisebox{-0.02cm}
{\tiny $\wedge$}\hspace{-0.121cm}\raisebox{0.08cm}{\tiny $.$}\hspace{-0.079cm}\raisebox{0.10cm}
{\tiny $.$}\hspace{-0.079cm}\raisebox{0.12cm}{\tiny $.$}\hspace{-0.085cm}\raisebox{0.14cm}
{\tiny $.$}\hspace{-0.079cm}\raisebox{0.16cm}{\tiny $.$}\hspace{1pt}}}
\newcommand{\oo}{\bot}            
\newcommand{\pp}{\top}            
\newcommand{\hint}{\mbox{\bf Int\hspace{-2pt}}^{\intimpli}}
\newcommand{\hintp}{{\mbox{\bf Int\hspace{-2pt}}}^{\pintimpli}}
\newcommand{\hintb}{\mbox{\bf Int\hspace{-2pt}}^{{\intimpli}^{\aleph_0}}}

\newcommand{\pst}{\mbox{\raisebox{-0.01cm}{\scriptsize $\wedge$}\hspace{-4pt}\raisebox{0.16cm}{\tiny $\mid$}\hspace{2pt}}}
\newcommand{\gneg}{\neg}                  
\newcommand{\intimpl}{\mbox{\hspace{2pt}$\circ$\hspace{-0.14cm} \raisebox{-0.058cm}{\Large --}\hspace{2pt}}}
\newcommand{\bintimpl}{\intimpl^{\aleph_0}}
\newcommand{\intimpli}{\mbox{\hspace{2pt}$\circ$\hspace{-0.14cm} \raisebox{-0.031cm}{\Large -}\hspace{2pt}}}
\newcommand{\mli}{\rightarrow}                     
\newcommand{\mld}{\vee}    
\newcommand{\mlc}{\wedge}  
\newcommand{\ade}{\mbox{\Large $\sqcup$}}      
\newcommand{\ada}{\mbox{\Large $\sqcap$}}      
\newcommand{\add}{\sqcup}                      
\newcommand{\adc}{\sqcap}                      
\newcommand{\sti}{\mbox{\raisebox{-0.02cm}
{\scriptsize $\circ$}\hspace{-0.121cm}\raisebox{0.08cm}{\tiny $.$}\hspace{-0.079cm}\raisebox{0.10cm}
{\tiny $.$}\hspace{-0.079cm}\raisebox{0.12cm}{\tiny $.$}\hspace{-0.085cm}\raisebox{0.14cm}
{\tiny $.$}\hspace{-0.079cm}\raisebox{0.16cm}{\tiny $.$}\hspace{1pt}}}
\newcommand{\costi}{\mbox{\raisebox{0.08cm}
{\scriptsize $\circ$}\hspace{-0.121cm}\raisebox{-0.01cm}{\tiny $.$}\hspace{-0.079cm}\raisebox{0.01cm}
{\tiny $.$}\hspace{-0.079cm}\raisebox{0.03cm}{\tiny $.$}\hspace{-0.085cm}\raisebox{0.05cm}
{\tiny $.$}\hspace{-0.079cm}\raisebox{0.07cm}{\tiny $.$}\hspace{1pt}}}
\newcommand{\intf}{\$}               
\newcommand{\st}{\mbox{\raisebox{-0.05cm}{$\circ$}\hspace{-0.13cm}\raisebox{0.16cm}{\tiny $\mid$}\hspace{2pt}}}
\newcommand{\bst}{\st^{\aleph_0}}
\newcommand{\cost}{\mbox{\raisebox{0.12cm}{$\circ$}\hspace{-0.13cm}\raisebox{0.02cm}{\tiny $\mid$}\hspace{2pt}}}
\newcommand{\pcost}{\mbox{\raisebox{0.12cm}{\scriptsize $\vee$}\hspace{-4pt}\raisebox{0.02cm}{\tiny $\mid$}\hspace{2pt}}}


\newtheorem{theoremm}{Theorem}[section]
\newtheorem{factt}[theoremm]{Fact}
\newtheorem{thesiss}[theoremm]{Thesis}
\newtheorem{definitionn}[theoremm]{Definition}
\newtheorem{figuree}[theoremm]{Figure}
\newtheorem{lemmaa}[theoremm]{Lemma}
\newtheorem{propositionn}[theoremm]{Proposition}
\newtheorem{conventionn}[theoremm]{Convention}
\newtheorem{examplee}[theoremm]{Example}
\newtheorem{remarkk}[theoremm]{Remark}
\newtheorem{conjecturee}[theoremm]{Conjecture}
\newtheorem{claimm}[theoremm]{Claim}
\newtheorem{corollaryy}[theoremm]{Corollary}

\newenvironment{definition}{\begin{definitionn} \em}{ \end{definitionn}}
\newenvironment{theorem}{\begin{theoremm}}{\end{theoremm}}
\newenvironment{lemma}{\begin{lemmaa}}{\end{lemmaa}}
\newenvironment{corollary}{\begin{corollaryy}}{\end{corollaryy}}
\newenvironment{proposition}{\begin{propositionn} }{\end{propositionn}}
\newenvironment{convention}{\begin{conventionn} \em}{\end{conventionn}}
\newenvironment{remark}{\begin{remarkk} \em}{\end{remarkk}}
\newenvironment{proof}{ {\bf Proof.} }{\  $\Box$ \vspace{.1in} }
\newenvironment{idea}{ {\bf Proof idea.} }{\  $\Diamond$ \vspace{.1in} }
\newenvironment{conjecture}{\begin{conjecturee} }{\end{conjecturee}}
\newenvironment{claim}{\begin{claimm} }{\end{claimm}}
\newenvironment{thesis}{\begin{thesiss} }{\end{thesiss}}
\newenvironment{example}{\begin{examplee} \em}{\end{examplee}}
\newenvironment{figuret}{\begin{figuree} \em}{\end{figuree}}
\newenvironment{fact}{\begin{factt} \em}{\end{factt}}

\title{Many concepts and two logics of algorithmic reduction}
\author{Giorgi Japaridze
\\ \\ Institute of Artificial Intelligence, Xiamen University and
\\ Department of Computing Sciences, Villanova University} 
\date{}
\maketitle

\begin{abstract} Within the program of finding axiomatizations for various parts of {\em computability logic}, it was  proven earlier that the logic of interactive Turing reduction is exactly the implicative fragment of Heyting's intuitionistic calculus. That sort of reduction permits unlimited reusage of the computational resource represented by the antecedent.  An at least equally basic and natural sort of algorithmic reduction, however, is the one that does not allow such reusage. The present article shows that turning the logic of the first sort of reduction  into the logic of the second sort of reduction takes nothing more than just deleting the contraction rule from its Gentzen-style axiomatization. The first (Turing) sort of interactive reduction is also shown to come in three natural versions. While those three versions are very different from each other, their logical behaviors (in isolation) turn out to be indistinguishable, with that common behavior being precisely captured by implicative intuitionistic logic. Among the other contributions of the present article is an informal introduction of a series of new --- finite and bounded --- versions of recurrence operations and the associated reduction operations.   
\end{abstract}

\noindent {\em MSC}: primary: 03B47; secondary: 03F50; 03B70; 68Q10; 68T27; 68T30; 91A05 


\noindent {\em Keywords}: Computability logic; Intuitionistic logic; Affine logic; Linear logic; Interactive computation; Game semantics.

\section{Introduction}\label{sintr}

This article is a new addition to the evolving list of papers \cite{Japtocl1,Japtocl2,Japtcs,Cirq,Japjsl,Japtcs2,int1,Propint,deep,seq,ar,Ver}
devoted to finding axiomatizations for various fragments of {\em computability logic}. The latter is a  program for redeveloping logic as a formal theory of computability, as opposed to a formal theory of truth which it has more traditionally been. 

Under the approach of computability logic, formulas express interactive computational problems defined as games between the two players $\pp$ ({\em machine}) and $\oo$ ({\em environment}), with logical operators standing for basic operations on games. 
``Truth'' of a problem/game means existence of an algorithmic solution, i.e. $\pp$'s effective winning strategy. And validity of a logical formula is understood as (such) truth under every particular interpretation of atoms. 
 With this semantics, computability logic provides a systematic answer to the fundamental question ``{\em what can be computed?}\hspace{1pt}'', just as classical logic is a systematic tool for telling what is true. Furthermore, as it turns out, in positive cases ``{\em what} can be computed'' always allows itself to be replaced by ``{\em how} can be computed'', which makes computability logic of potential interest in not only theoretical computer science, but many more applied areas as well, including interactive knowledge base systems, resource oriented systems for planning and action, or declarative programming languages. 
On the logical side, computability logic can serve as a constructive and computationally meaningful alternative to classical logic as a basis for applied theories.
The first concrete steps in the direction of materializing this potential have been made very recently in \cite{ar}, where a computability-logic-based system of arithmetic was constructed --- a formal theory whose every formula expresses a computational problem and every proof encodes an algorithmic solution for such a problem, thus fully reducing problem-solving to theorem-proving. 

Having said the above, motivationally or technically (re)introducing computability logic is not within the scope of the present paper. This job has been done in  
 \cite{Jap03,Japic,Japfin}, and the present paper, whose goal is merely putting one more brick into the foundation of the edifice under construction, 
primarily targets readers already familiar with the basics of computability logic. Yet, as it happens, the proof of the main technical result of the paper, given in Sections 2-4, can be understood in full detail without knowing much (if anything at all) about computability logic. Those with no prior acquaintance with the subject may  benefit from browsing the rest of the paper just as well. Even though doing so would be certainly insufficient for getting full insights into the project, chances are that such a reader may at least start feeling curious enough to be willing to look at some additional literature. The most recommended reading for familiarity with the basic philosophy, motivations, concepts and techniques of computability logic is the tutorial-style \cite{Japfin}. 

Here we very quickly review, in a simplified form, certain basic concepts on the games used in computability logic, to refresh the memory of those previously exposed to the subject, and to provide some clues to those who have never seen it. 

A {\bf move} means a finite string over some fixed alphabet, such as the standard keyboard alphabet. A {\bf labeled move} is a move prefixed with $\pp$ or $\oo$. The meaning of such a prefix (``{\em label}'') is to indicate which of the two players has made the move. A {\bf run} is a (finite or infinite) sequence of labeled moves, and a {\bf position} is a finite run. Runs (and positions as special cases of runs) are thus records of interaction histories, spelling out what moves, in what order and by which players have been made during a given play of a game. 

A {\bf game}\footnote{To what we refer as a ``game'' in this paper, is in fact called a ``constant game'' in computability logic, and the term ``game'' is reserved for a slightly more general concept. Considering only constant games is sufficient for our present purposes though and, to
 keep things simple, we are using the term ``game'' for them.} is a pair $(\mbox{\bf Lr},\mbox{\bf Wn})$ consisting of what are called its {\bf structure} ({\bf Lr}) and {\bf content} ({\bf Wn}). One of the many equivalent ways to define the structure component of a game is to 
say that it is a binary relation between positions and labeled moves. Then the intuitive meaning of $\mbox{\bf Lr}(\Phi,\xx\alpha)$ is that $\alpha$ is a {\bf legal move} by player $\xx$ in position $\Phi$. A run where all moves are legal (in the positions preceding those moves) is 
said to be a {\bf legal run}. The empty run is thus always trivially legal. As expected, ``{\bf illegal}'', whether it be a move or a run, means ``not legal''. As for the content {\bf Wn} of a game, it can be defined as a set of legal runs, whose elements are said to be (and 
intuitively thought of as) the runs {\bf won} by player $\pp$ (and hence lost by $\oo$), with all other legal runs considered {\bf lost} by $\pp$ (and hence won by $\oo$). As for illegal runs, they are always considered to be lost by the player who made the first illegal move. 

Note the relaxed nature of such games. There are no conditions on the order in which moves should or could be made (such as, say, strict alternation of players'  turns), and generally either player may have legal moves in a given position/situation. This makes the games of computability logic a rather direct (without any ``bureaucratic pollutants'') and flexible tool for modeling interaction, including asynchronous interactions. The relaxed nature of our games  makes it impossible to understand game-playing strategies as functions from positions to moves, as this is typical for most other game models. Instead, ($\pp$'s effective) strategies are understood as interactive machines. Such a machine is nothing but a Turing machine with the additional capability of making moves. The adversary can also move at any time, with such moves being the only nondeterministic events from the machine's perspective. The play is fully visible to the machine through an additional, read-only {\em run tape} which, at any time, spells the ``current position'' of the play. We say that such a machine {\bf wins} a given game iff, no matter how the adversary acts (what moves it makes and when it makes them), the run incrementally spelled on the run tape is won by $\pp$.

A universal-utility game semantics should be  about interaction, whereas functions are inherently non-interactive. The above-mentioned traditional, {\em strategies-as-functions}, approach misses this important point and creates a hybrid of interactive (games) and non-interactive (functions) entities. To see the resulting loss, it would be sufficient to reflect on the behavior of one's personal computer. The job of your computer is to play one long --- potentially infinite --- game against you. Now, have you noticed your ``adversary'' getting slower every time you use it? Probably not. That is because the computer is smart enough to follow a non-functional strategy in this game. If its strategy was a function from positions (interaction histories) to moves, the response time would inevitably keep worsening due to the need to read the entire --- continuously lengthening and, in fact, practically infinite --- interaction history every time before responding. Defining strategies as functions of only the latest moves (rather than entire interaction histories) in Abramsky and Jagadeesan's \cite{Abr94} tradition  is also not a way out, as typically more than just the last move matters. Back to your personal computer, its actions certainly depend on more than 
your last keystroke. Thus, the difference between the traditional {\em functional} strategies and the {\em post-functional} strategies of computability logic is not just a matter of taste or convenience. It will become especially important when it comes to (yet to be developed) interactive complexity theory: hardly any meaningful interactive complexity theory can be done with the strategies-as-functions approach. And complexity issues will inevitably come forward when computability logic or similar approaches achieve a certain degree of maturity: nowadays, 95\% of the theory of computation is about complexity rather than just computability.
Time has not yet matured for seriously addressing complexity issues within the framework of computability logic though, and the latter, including the present paper, continues to be focused on just computability, which still abounds with open questions waiting for answers.  

In the above outline, we described interactive Turing machines in a relaxed fashion, leaving to the reader filling technical details about, say, how, exactly, moves are made by the machine, how many moves either player can make at once, what happens if both players attempt to move ``simultaneously'', etc. As it turns out, all reasonable design choices yield the same class of winnable games as long as we consider a certain natural subclass of games called {\bf static}. Such games are obtained by imposing a certain simple formal condition on games (see, e.g., Section 5 of \cite{Japfin}), which we do not reproduce here as nothing in this paper relies on it. We will only point out that, intuitively, static games are interactive tasks where the relative speeds of the players are irrelevant, as it never hurts a player to postpone making moves. In other words, static games are the games that are contests of intellect rather than contests of speed. And one of the theses that computability logic philosophically relies on is that static games present an adequate formal counterpart of our intuitive concept of ``pure'', speed-independent interactive computational problems. Correspondingly, computability logic restricts its attention (more specifically, possible interpretations of the atoms of its formal language) to static games. Needless to say, the class of static games is closed under all game operations studied in computability logic.

Among the most interesting of such operations are several versions of {\em reduction}. 
The  simplest form of reduction (of $B$ to $A$) is $A\mli B$. This is a parallel play of the two games $A$ and $B$ with the roles of $\pp$ and $\oo$ 
interchanged in the antecedent. Winning a given run of this game  for $\pp$ means that whenever the adversary wins $A$, $\pp$ has to win $B$.

More formally, every legal move of $A\mli B$ has to be prefixed with one of the two strings ``$0.$'' or ``$1.$'' to indicate in which of the two components the move is made. The effect of a move $0.\alpha$ is making move $\alpha$ in the antecedent, and the effect of $1.\alpha$ is making move $\alpha$ in the consequent. In order for such a move $0.\alpha$ or $1.\alpha$ to be legal, $\alpha$ should be a legal move in (the corresponding position of) the corresponding component $A$ or $B$. Then a legal run $\Gamma$ of $A\mli B$ is considered won by $\pp$ iff $\Gamma^{1.}$ is a $\pp$-won run of $B$ or $\gneg\Gamma^{0.}$ is a $\oo$-won run of $A$. Here $\Gamma^{1.}$ means the result of deleting from $\Gamma$ all moves except those of the form $1.\alpha$, and then further deleting the prefix ``$1.$'' in such moves. Similarly for $\Gamma^{0.}$. And $\gneg\Gamma^{0.}$ means the result of turning upside down (so that $\pp$ becomes $\oo$ and vice versa) all labels in $\Gamma^{0.}$. 

As can be felt from the above passage, formal definitions may not be as nice to work with as informal or intuitive explanations. For this reason, 
our subsequent explanations of game operations in this section will be limited to informal ones, keeping in mind that they certainly can be turned into strict technical definitions.     

Since the roles of the players are switched in the antecedent of $A\mli B$, the $A$ component, as a computational problem from $\oo$'s perspective, becomes a computational resource for $\pp$. Namely, $\pp$ can observe how the adversary is solving/playing (a single session) of $A$, and utilize that information in its own solving/playing $B$. 
The following example illustrates the above-said. Let
$\mbox{\bf H}$ be the {\em halting problem}, which can be understood as a game of depth $2$ (i.e., no legal run has more than two moves). In the initial (empty) position of this game, only $\oo$ has legal moves, and such a move should be the phrase ``Does Turing machine $m$ halt on input $i$?'', where $m$ is a legitimate description of a Turing machine and $i$ a possible input for it. After such a move is made, the second and last legal move is by $\pp$, which should be either ``Yes'' or ``No''. $\pp$ wins iff it correctly answers the question asked by $\oo$. The failure by $\oo$ to make an initial move is considered $\pp$'s win, as there was no question to answer. And, if such a move is made, then the failure of $\pp$ to respond is considered $\oo$'s win. The {\em acceptance problem} {\bf A} is similar, only it is about whether a given machine accepts (rather than halts on) a given input. Neither {\bf H} nor {\bf A} is decidable, which obviously means that these problems, as games, have no algorithmic winning strategies. However, {\bf A} is algorithmically reducible to {\bf H}. Specifically, $\pp$ {\em does} have an effective winning strategy in the game $\mbox{\bf H}\mli \mbox{\bf A}$, which goes like this. Wait till, in the consequent,  $\oo$ asks a question regarding whether a certain machine $m$ accepts a certain input $i$. Then, in the antecedent, ask a counterquestion regarding whether $m$ halts on $i$ (the same $m$ and $i$). If an answer to this counterquestion is ``No'', answer ``No'' to the original question in the consequent and rest your case, as not halting implies not accepting. Otherwise, if the answer in the antecedent is ``Yes'', simulate machine $m$ on input $i$ until it halts, and say ``Yes'' or ``No'' in the consequent depending on whether the simulation accepted or rejected. (Of course, the possibility that the simulation goes  on forever is not ruled out here; but this would mean that $m$ does not really halt on $i$, and having lied in the antecedent would make $\oo$ lose the game regardless of what happens in the consequent).

In fact, $\mli$ is not only the simplest but also the strongest form of reduction. In this respect, at  the other extreme is the weakest reduction $\intimpl$. 
The game $A\intimpl B$ can be characterized in the same intuitive terms as $A\mli B$, with the difference that, in $A\intimpl B$, unlike $A\mli B$, $\pp$ is allowed to reuse $A$ (as a computational resource) any number of times, with ``reuse'' here understood in the strongest algorithmic sense possible. Namely, at any time, $\pp$ can temporarily abandon  a given position of $A$ (while reserving the right to come back to it later), backtrack to any of the earlier positions of it and try a different continuation from there, thus forcing $\oo$ to play multiple parallel sessions of $A$ against such a capricious adversary in this most unfair game: 
the failure of $\oo$ to win $A$ in {\em all} sessions of it automatically results in $\pp$'s victory.  

A while ago we saw how to reduce the acceptance problem to the halting problem in the strong sense of $\mli$. We would not have been just as successful if instead of the acceptance problem {\bf A} we had taken the {\em Kolmogorov complexity} problem {\bf K}, where the initial move of the form ``What is the Kolmogorov complexity of number $n$?'' is by the environment, obligating the machine to respond with a move ``$m$'' such that $m$ is (indeed) the Kolmogorov complexity of $n$, i.e., $m$ is the size of the smallest Turing machine that returns $n$ on input $0$. One can show that, unlike $\mbox{\bf H}\mli\mbox{\bf A}$, the game $\mbox{\bf H}\mli\mbox{\bf K}$ does not have an algorithmic winning strategy. But the weaker game $\mbox{\bf H}\intimpl\mbox{\bf K}$ certainly does, due to the fact that, in it, the reduction is allowed to use the antecedent repeatedly. Such a strategy goes like this. Wait to hear a question about the Kolmogorov complexity of a number $n$ in the consequent. Then, starting from $m=0$, do the following. Duplicate the original antecedent and save one copy of it for future usage (further duplications). In the other copy, ask the counterquestion regarding whether the machine (encoded by) $m$ halts on input $0$. If you hear ``No'', increment $m$ to $m+1$ and repeat the step. Otherwise, if you hear ``Yes'', simulate $m$ on input $0$; if the simulation shows that $m$ returns $n$ on input $0$, answer $|m|$ (where $|m|$ is the size of $m$) to the original question in the consequent, and wash your hands. In any other case,  increment $m$ to $m+1$ and repeat the step.

There is a whole spectrum of natural reduction operations of intermediate strength between $\mli$ and $\intimpl$. Only some of those have been officially introduced within the framework of computability logic so far, with more to be probably defined later depending on particular needs, motivations and tastes. It has been repeatedly pointed out earlier that the formalism of computability logic is open-ended, welcoming any meaningful augmentations. 

Among the most natural and simple reduction operations of intermediate strength  is $\pintimpl$. Just like $A\intimpl B$ and unlike $A\mli B$, the game $A\pintimpl B$ allows $\pp$ to reuse $A$ infinitely many times. But the form of reusage is less flexible here:  $\pp$ is required to restart $A$ from the very beginning every time it wants to reuse it, meaning that  it essentially cannot utilize the advantages of backtracking permitted in $A\intimpl B$. Specifically, unless 
$\oo$ plays in exactly the same ways in different parallel sessions of $A$, $\pp$ has no possibility to experiment with different reactions to the same actions by the adversary.

Thus, the difference between $A\pintimpl B$ and $A\intimpl B$ is in the allowed {\em type} of reusage of $A$, with the {\em quantity} of reusages being otherwise unlimited. Yet, as it happens, this difference in the types of reusage automatically yields a difference in the quantities as well. Specifically, in $A\pintimpl B$ at most countably many parallel runs of $A$ can be generated, while in $A\intimpl B$, when $A$ has infinitely long legal runs, that quantity can be a continuum. A very simple modification in the formal definition of $\intimpl$, given later in Section \ref{s5}, turns it into a definition of the Blass-style 
\cite{Bla72} reduction $\bintimpl$, by its strength strictly between $\pintimpl$ and $\intimpl$. The type of reusage of $A$ in $A\intimpl^{\aleph_0} B$ is the same as in $A\intimpl B$, but the quantity of reusages is limited to the countably infinite cardinal $\aleph_0$. 

The operation $\intimpl^{\aleph_0}$ is apparently the weakest nontrivial strengthening of $\intimpl$. Both $\intimpl$ and $\pintimpl$ can be further  strengthened to $\intimpl^F$ and $\pintimpl^F$ by allowing $\pp$  to reuse the antecedent only a finite (yet unbounded) number of times. In turn, the operations $\intimpl^F$ and $\pintimpl^F$ can be further strengthened to bounded versions of $\intimpl,\pintimpl$. The simplest form of a bounded version of $\supset\in\{\pintimpl,\intimpl\}$ would be $\supset^n$, where $n$ is a natural number. It means the same as $\supset$, only the number of allowed (re)usages of the antecedent is limited to $n$, so that $A\supset^0\hspace{-3pt} B$ is nothing but simply $B$, and 
$A\supset^1 \hspace{-3pt} B$ is nothing but $A\mli B$. But bounds do not necessarily have to be natural numbers. Reasonable transfinite ordinals could be interesting to study as well, such as ordinals less than $\epsilon_0$. For example, where $\omega$ is the smallest infinite ordinal, $A\supset^{\omega}\hspace{-3pt} B$ would mean a game where $\pp$ has to declare a number $n$ before starting using $A$, after which the game continues as if it was $A\supset^n \hspace{-3pt}B$. This generalizes to  $A\supset^{k\omega}\hspace{-3pt} B$ for any $k\geq 0$, where $k$ (rather than just one)
declarations $n_1,\ldots,n_k$ are made. The first declaration $n_1$ opens $n_1$ copies of $A$ for usage; the second declaration $n_2$, which can be made any time later when the previously ``activated'' copies of $A$ are perhaps already at advanced stages, creates the possibility to use $n_2$ additional copies; the third declaration activates $n_3$ additional copies, etc., with the overall number of (re)usages of $A$ thus eventually not exceeding the finite $n_1+\ldots+n_k$. 
 Next,  $A\supset^{\omega^2}\hspace{-5pt} B$ would be a game where $\pp$ has to declare a number $n$ before starting using $A$, after which the play continues as it would proceed in $A\supset^{n\omega}\hspace{-3pt} B$. This can be further generalized to $A\supset^{\omega^k}\hspace{-4pt} B$ for any $k\geq 0$. Then $A\supset^{\omega^\omega} \hspace{-4pt} B$ could be characterized as 
a game where $\pp$'s initial choice of $n$ turns it into a game that proceeds as $A\supset^{\omega^n}\hspace{-4pt} B$. And so on and so on. 

Furthermore, $\intimpl$ has an even greater variety of bounded versions of potential interest, especially in the (yet to be developed) area of  interactive computational complexity theory. One may want to differentiate between just bounds on the overall number of reusages of the antecedent and bounds on, say, the ``depths'' of reusages. Roughly, 
the depth of reusages here means the maximum number of ancestor positions of any given run of the antecedent at which restarts (``forkings'', ``replications'') happened. In more precise terms --- for those familiar with the relevant formal definitions --- such bounds would mean bounds on the heights of the 
corresponding underlying bitstring trees (see \cite{Japfin}). For $\pintimpl$, on the other hand, the above concept of depth is not meaningful as it automatically trivializes to $1$ (or to $0$, depending on whether or not only proper reusages count). 

Finite or bounded versions of reduction operations, except the ``most finite'' and ``most bounded'' $\mli$, have never been studied, and at this point we do not know what logics they induce. In what follows our focus is only on $\mli, \pintimpl, \intimpl^{\aleph_0},\intimpl$.

Of these four operations, $\intimpl$ stands out as, in a sense, most natural and important. What makes $\intimpl$ special is that it has good claims to precisely capture everything that anyone would ever call (interactive) algorithmic reduction. That is in the same sense as Turing computability 
of functions captures our intuitive concept of effectiveness.  What also makes $\intimpl$ natural is that, as suggested by the above characterizations, definitions of other reductions  can be easily obtained from the definition of $\intimpl$ by imposing corresponding restrictions on the form and quantity of reusage of the antecedent, with $\mli$ being the most extreme nontrivial case, where any proper reusage is simply forbidden altogether. 

Alternatively, we can consider $\mli$ rather than $\intimpl$ as the basic sort of reduction, and define all weaker versions of reduction in terms of $\mli$ and what are called {\em recurrence operations} ($\pst,\st^{\aleph_0},\st,\ldots$), in their general spirit resembling the storage operator $!$ of linear logic. This is exactly the approach that computability logic has prefered to take so far.\footnote{In fact, computability logic further decomposes $\mli$, defining $A\mli B$ as $\gneg A\mld B$.} 
For instance,  \cite{Japfin} treats $A\pintimpl B$ and $A\intimpl B$ as abbreviations of $\pst A\mli B$ and $\st A\mli B$, respectively. This sort of a decomposition of weak implication-style operators looks well familiar from linear logic \cite{Gir87}, or the even earlier work \cite{Bla72} by Blass. So, the above discussion of various new sorts of reduction can be in fact considered an informal introduction of the corresponding series of new recurrence operations. For this reason, and also for the (related) reason of being the only fully resource-conscious reduction, the operation $\mli$ is at least as important, basic and natural as $\intimpl$. 

The operations $\pintimpl$, $\bintimpl$ and $\intimpl$ equally enjoy the status of being conservative generalizations of Turing reduction for the interactive context. Specifically, when $A$ and $B$ are traditional sorts of problems such as decision problems or problems of computing a function, effective winnability of any of the three   games $A\pintimpl B$,  $A\bintimpl B$, $A\intimpl B$ turns out to coincide with Turing reducibility of $B$ to $A$, with the subtle differences between $\pintimpl,\bintimpl,\intimpl$ becoming relevant only when these operators are applied to problems with higher degrees of interactivity. The same does not extend to $A\mli B$ though:
(even) when $A$ and $B$ are traditional sorts of problems, effective winnability of $A\mli B$ means something properly stronger. It means existence of a Turing machine that solves $B$ with an oracle for $A$ where the oracle can be queried only once (while, as we know, ordinary Turing reducibility does not impose any limits on how many times the oracle can be queried). The earlier mentioned finite versions $\pintimpl^F$ and $\intimpl^F$ of weak reductions can also be seen to be conservative generalizations of Turing reduction. As for the bounded versions of weak reductions, they generalize certain proper strengthenings of Turing reduction, obtained by imposing (finite or transfinite) bounds on the number of possible queries of the oracle. Going back to our Kolmogorov complexity example, the game $\mbox{\bf H}\supset \mbox{\bf K}$ has an algorithmic winning strategy for each $\supset\in\{\pintimpl,\intimpl,\bintimpl,\pintimpl^F,\intimpl^F\}$. Furthermore, with some thought and keeping in mind the known fact that the Kolmogorov complexity of $n$ never exceeds  $n$ itself (for the exception of a finite number of ``very small'' $n$'s), one can see that $\mbox{\bf H}\supset \mbox{\bf K}$ remains algorithmically solvable with either $\supset\in\{\pintimpl^\omega,\intimpl^\omega\}$ as well.  

As it turns out, the logical behaviors of $\pintimpl$, $\bintimpl$ and $\intimpl$ are indistinguishable when these operators are taken in isolation, and that common behavior is precisely captured by the implicative fragment  $\mbox{\bf Int}^\supset$ of Heyting's intuitionistic calculus. For $\pintimpl$ and $\intimpl$, a proof of this fact was given in \cite{Japjsl}. And the present paper extends that result to $\bintimpl$ as well.
As for $\mli$, it turns out that its logical behavior is captured by {\bf CL7}, which is (the Gentzen-style axiomatization of) ${\bf Int}^\supset$ with just the contraction rule deleted. In other words, {\bf CL7} is nothing but the implicative fragment of affine logic. 
A proof of this result is the main technical contribution of the present article. And this is not a result that could be taken for granted.  As shown in 
\cite{Japfin}, affine logic in its full language, while sound, is far from being complete with respect to the semantics of computability logic. In fact, even just the 
$(\mli,\gneg)$-fragment of  computability logic is not the same as the corresponding fragment of affine logic, nor does it appear to be axiomatizable in 
traditional proof theory. 

Another way to summarize the main technical result of the present paper is to say that the set of implicative binary tautologies and their substitutional instances is precisely described by {\bf CL7}. Here {\em binary tautologies} mean tautologies of classical logic where no atom occurs more than twice, and 
 {\em implicative binary tautologies} are  binary tautologies that contain no connectives others than $\mli$. Binary tautologies and their instances have arisen in the past as a class of formulas sound and complete with respect to several natural semantics, most notably Blass's game semantics for linear logic \cite{Bla92}, Blass's resource-conscious semantics for classical logic \cite{Bla03}, the semantics of computability logic \cite{Jap03}, and abstract resource semantics \cite{Cirq,deep}. This class of formulas has stubbornly resisted any axiomatization attempts within the framework of traditional deductive approaches and,  as argued by Blass in \cite{Bla92}, apparently this phenomenon is not quite an accident. 
A reasonable axiomatization for the set of binary tautologies and their instances was eventually found in \cite{Cirq}, but it took switching to a substantially new deductive framework called {\em cirquent calculus} (roughly, it is sequent calculus where formulas may be shared between different sequents), indirectly corroborating Blass's thesis that binary tautologies are foreign to traditional proof theory.  Against this background, the fact that the implicative fragment of that wild class can still be tamed with traditional means such as substructural sequent calculus in which {\bf CL7} is constructed, is worth receiving our attention. 

\section{Logic {\bf CL7}}\label{s2}
The languages that we consider in this paper have infinitely many nonlogical propositional atoms for which we use the metavariables $P,Q$, and have no logical atoms. Where $\supset\in\{\mli,\pintimpl,\bintimpl,\intimpl\}$, by a {\bf $\supset$-formula} we mean a formula built from atoms and (the binary) $\supset$ in the standard way. We will be using $E,F,G,H$ as metavariables for formulas, and $\Gamma,\Delta$ as metavariables for (possibly empty) multisets of formulas. As usual, we write $\Gamma,\Delta$ or $\Gamma,F$ instead of  $\Gamma\cup \Delta$  or $\Gamma\cup\{F\}$. A (two-sided) {\bf $\supset$-sequent} is 
a pair $\Gamma\Rightarrow F$, where $\Gamma$ is a finite multiset of $\supset$-formulas and $F$ is a $\supset$-formula. Here $\Gamma$ is said to be the {\bf antecedent} of the sequent, and $F$ is said to be the {\bf succedent}.

We axiomatize {\bf CL7} using two-sided $\mli$-sequents. A \mbox{($\mli$-)} formula $H$ is considered provable in this system (written \mbox{${\bf CL7}\vdash H$)} iff the empty-antecedent sequent $\Rightarrow H$ is so. 

The {\bf axioms} of {\bf CL7} are all $\mli$-sequents of the form \[\Gamma,\hspace{2pt}F\ \Rightarrow \ F.\]
And the system only has the following two {\bf rules of inference}:

\begin{center}
\begin{picture}(280,70)
\put(223,30){$\Gamma,\hspace{2pt}E\ \Rightarrow \ F$}
\put(227,50){\bf Right $\mli$}
\put(220,22){\line(1,0){55}}
\put(220,8){$\Gamma\ \Rightarrow\  E\hspace{-2pt}\mli\hspace{-3pt} F$}

\put(0,30){$\Gamma,\hspace{2pt}F\ \Rightarrow \ G$}
\put(86,30){$\Delta\ \Rightarrow \ E$}
\put(50,50){\bf Left $\mli$}
\put(0,22){\line(1,0){122}}
\put(18,8){$\Gamma,\hspace{2pt}\Delta,\hspace{2pt}E\hspace{-2pt}\mli\hspace{-3pt} F\ \Rightarrow \ G$}
\end{picture}
\end{center}

We say that a formula of classical propositional logic (with $\mli$-formulas here also seen as such) is {\bf binary} iff 
no atom occurs in it more than twice. The concepts of being binary, tautological, true or false extend from formulas to sequents by understanding each sequent $E_1,\ldots,E_n\Rightarrow F$ as the formula $E_1\mlc\ldots \mlc E_n\mli F$. A (substitutional) {\bf instance} of a given formula $F$, as usual, means the result of replacing atoms in $F$ by any formulas, with all occurrences of the same atom being replaced by the same formula, of course.   

\begin{theorem}\label{th}
For any $\mli$-formula $H$, the following conditions are equivalent:

(i) \ \ ${\bf CL7}\vdash H$.

(ii) \ $H$ is an instance of a binary tautology.

(iii) $H$ is valid in computability logic, whether it be in the ordinary sense of validity or in the stronger sense of 
 what is called ``uniform validity'' (see \cite{Japfin}).   
\end{theorem}

\begin{proof} The equivalence between (ii) and (iii) in a stronger form which is not restricted to just $\mli$-formulas, has been proven in \cite{Cirq}.\footnote{A game-semantical soundness and completeness of the class of substitutional instances of binary tautologies was first proven with respect to Blass's game semantics in \cite{Bla92}.} So, to prove the present theorem, it would be sufficient to show that (i) implies (iii) (call this {\em soundness}) and that (ii) implies (i) (call this {\em completeness}). This will be done in the following two sections. 
\end{proof}

\section{The soundness of {\bf CL7}}
We can rewrite {\bf CL7} into a clearly equivalent system that uses {\bf one-sided sequents}, here restricted to finite multisets of formulas of classical propositional logic without $\mli$, where negation is applied only to atoms. This is done by rewriting each $\mli$-sequent $E_1,\ldots,E_n\Rightarrow F$ as $\gneg E_1,\ldots,\gneg E_n,F$, and then iteratively rewriting each (sub)formula $E\mli F$ as $\gneg E\mld F$, each subformula $\gneg(E\mld F)$ as $\gneg E\mlc\gneg F$, each subformula $\gneg (E\mlc F)$ as $\gneg E\mld\gneg F$ and each subformula $\gneg\gneg E$ as $E$. The axioms of the resulting system are all sequents of the form $\Gamma,\gneg F,F$,\footnote{Of course, it does not matter whether here and later we write $\Gamma$ or $\gneg \Gamma$, with $\gneg \Gamma$ meaning the multiset of the negations of the elements of $\Gamma$.} and the rules of inference now read as follows:

\begin{center}
\begin{picture}(280,70)
\put(228,30){$\Gamma,\hspace{2pt}\gneg E,\hspace{2pt} F$}
\put(227,50){\bf Right $\mli$}
\put(220,22){\line(1,0){55}}
\put(225,8){$\Gamma,\hspace{2pt}\gneg E\hspace{-1pt}\mld\hspace{-1pt} F$}

\put(0,30){$\Gamma,\hspace{2pt}\gneg F,\hspace{2pt} G$}
\put(100,30){$\Delta,\hspace{2pt} E$}
\put(50,50){\bf Left $\mli$}
\put(0,22){\line(1,0){122}}
\put(25,8){$\Gamma,\hspace{2pt}\Delta,\hspace{2pt}E\hspace{-1pt}\mlc\hspace{-1pt} \gneg F,\hspace{2pt} G$}
\end{picture}
\end{center}

Among several equivalent axiomatizations of the (multiplicative fragment of the) well known {\em affine logic} is the one that uses one-sided sequents in our present sense. It has the same axiom scheme $\Gamma,\gneg F,F$. And the above Right $\mli$ and Left $\mli$ rules are special cases of the $\mld$-introduction and $\mlc$-introduction rules of that system, respectively, where $\mlc,\mld$ are seen as multiplicatives.\footnote{In fact, writing $E$ instead of $\gneg E$, Right $\mli$ is simply the same as the $\mld$-introduction rule of affine logic.} Thus, understanding $E\mli F$ as an abbreviation of $\gneg E\mld F$, affine logic proves every $\mli$-formula  provable in {\bf CL7}. But, as proven in \cite{Japfin}, affine logic is sound with respect to the semantics computability logic, and the latter sees no difference between $E\mli F$ and $\gneg E\mld F$. So, clause (i) of Theorem \ref{th} implies clause (iii), as desired.

\section{The completeness of {\bf CL7}}

We define the  {\bf head} of a $\mli$-formula as follows: 
\begin{itemize}
\item Every atom is its own head.
\item The head of $E\mli F$ is that of $F$.
\end{itemize}
In other words, the head of a formula is the atom with the rightmost occurrence in the formula --- the (unique) occurrence that is not in the antecedent of any subformula.  

Consider any binary $\mli$-sequent $\Gamma\Rightarrow F$. We define the {\bf relevant formulas} of this sequent to be the elements of the smallest set $S$ such that: 
\begin{itemize}
\item Every formula of $\Gamma$ whose head occurs in $F$ is in $S$.
\item Every formula of $\Gamma$ whose head occurs in some element of $S$ is also in $S$.
\end{itemize} 
The formulas of $\Gamma$ that are not relevant will be said to be {\bf irrelevant}.

\begin{lemma}\label{l1}
Assume $\Gamma\Rightarrow F$ is a binary tautological $\mli$-sequent, and $\Delta$ is the result of deleting from $\Gamma$ all irrelevant formulas of $\Gamma\Rightarrow F$. Then the sequent $\Delta\Rightarrow F$ is also tautological (and, of course, remains binary).
\end{lemma}

\begin{proof} Let $\Gamma,\Delta,F$ be as above. In what follows, by a ``relevant formula'' we always mean a relevant formula of $\Gamma\Rightarrow F$. Similarly for ``irrelevant'', ``antecedent'', ``succedent''.

Suppose that  
$\Delta\Rightarrow F$ is not tautological. Consider a truth assignment that makes it false, i.e., makes $\Delta$ true and $F$ false. Extend it to all formulas of $\Gamma$ by stipulating that, if an atom does not occur in $\Delta\Rightarrow F$, it is true. Obviously the head of every irrelevant formula is true under this extended assignment and hence every irrelevant formula is true. All relevant formulas of the antecedent also remain true. And the succedent remains false. So, $\Gamma\Rightarrow F$ is false and hence 
non-tautological.
\end{proof}

\begin{lemma}\label{l2}
Assume $\Gamma\Rightarrow E$ and $\Gamma\Rightarrow F$ are binary sequents, where $E$ and $F$ do not share any atoms.  Then the sets of relevant formulas of the two sequents are disjoint. 
\end{lemma}

\begin{proof} Assume the conditions of the lemma. Consider an arbitrary relevant formula $G$ of $\Gamma\Rightarrow E$. Let $P$ be the head of $G$. If the reason for $G$'s relevance is that 
$P$ occurs in $E$, then (as $E$ and $F$ share no atoms) $P$ does not occur in $F$, nor does it occur in any formula of $\Gamma$ other than $G$ because of the binarity of the sequent. This, by the definition of relevance, means that $G$ is not a relevant formula of 
 $\Gamma\Rightarrow F$.

Suppose now the reason for $G$'s being a relevant formula of $\Gamma\Rightarrow E$ is that $P$ occurs in some relevant formula $H$ of  $\Gamma\Rightarrow E$. The relevance of $H$ has thus been established earlier than that of $G$ and hence, by the induction hypothesis, $H$ is not a 
relevant formula of $\Gamma\Rightarrow F$. But, in view of binarity, the only two places where $P$ occurs (whether it be within $\Gamma\Rightarrow E$ or $\Gamma\Rightarrow F$) are in $G$ and $H$. Hence $G$ cannot be a relevant formula of  $\Gamma\Rightarrow F$. 
\end{proof}

\begin{lemma}\label{l3}
{\bf CL7} proves every binary tautological $\mli$-sequent. 
\end{lemma}

\begin{proof} Consider an arbitrary binary tautological sequent. We may assume that its succedent is an atom $P$, for otherwise, if the succedent is $E\mli F$, move $E$ to the antecedent of the sequent, and repeat the same until the succedent has become atomic; in view of the presence of Right $\mli$ in {\bf CL7}, provability of the resulting sequent implies provability of the original one. 

If $P$ is one of the formulas of the antecedent, then the sequent we deal with is an axiom and thus {\bf CL7} proves it.

Otherwise, the antecedent should contain a formula $E\mli F$ whose head is $P$, or else the sequent could be falsified by the truth assignment which makes $P$ false and makes all other atoms true. Thus, the sequent we are talking about looks like $\Gamma, E\mli F\Rightarrow P$, where $P$ occurs in $F$ and hence occurs in neither $E$ nor $\Gamma$, as the sequent is binary.  Obviously the tautologicity of this sequent implies the tautologicity of $\Gamma,F\Rightarrow P$. Since $E$ does not contain $P$, the tautologicity of $\Gamma, E\mli F\Rightarrow P$ also implies the tautologicity of $\Gamma \Rightarrow E$. Indeed, assume that some truth assignment falsifies $\Gamma \Rightarrow E$. Extend that assignment to all atoms of $\Gamma, E\mli F\Rightarrow P$ in such a way that it makes $P$ false. Obviously such an extended assignment falsifies $\Gamma, E\mli F\Rightarrow P$, contradicting our assumption that this sequent is tautological.  Thus, $\Gamma, F\Rightarrow P$ and $\Gamma\Rightarrow E$ are binary tautological sequents, and their succedents 
do not share any atoms. Let $\Gamma_1$ and $\Gamma_2$ be the submultisets of $\Gamma$ consisting of all relevant formulas of $\Gamma, F\Rightarrow P$ and $\Gamma\Rightarrow E$, respectively. By Lemma \ref{l2}, $\Gamma_1$ and $\Gamma_2$ are disjoint. Also, by Lemma \ref{l1}, $\Gamma_1, F\Rightarrow P$ and $\Gamma_2\Rightarrow E$ are tautological. Hence, by the induction hypothesis (where induction is on the number of connectives occurring the sequent), 
these two sequents are provable. Then, by Left $\mli$, the sequent $\Gamma_1,\Gamma_2, E\mli F\Rightarrow P$ is also provable. This can be easily seen to imply the provability of the original sequent  
$\Gamma, E\mli F\Rightarrow P$, as {\bf CL7} is obviously closed under the weakening rule ``from $\Delta\Rightarrow G$ conclude $\Delta,H\Rightarrow G$''.\footnote{In the present version of {\bf CL7}, weakening is ``hidden'' in axioms. Alternatively, we could have chosen the axioms of {\bf CL7} to be just $F\Rightarrow F$, with weakening explicitly stipulated as one of the inference rules. It is known that either choice yields the same set of provable formulas, whether it be classical, affine or intuitionistic logic.} 
\end{proof}

In view of the evident fact that {\bf CL7} is closed under substitution of atoms by whatever formulas, Lemma \ref{l3} immediately implies the desired conclusion that, whenever $H$ is a $\mli$-formula which is an instance of some binary tautology, $H$ is provable in {\bf CL7}. 

\section{The three versions of weak reduction}\label{s5}
As noted in Section 1, the three weak reduction operations $\pintimpl$, $\bintimpl$ and $\intimpl$ can be defined in terms of $\mli$ and the corresponding three {\em recurrence} operations 
$\pst,\bst,\st$ by 
\[A\pintimpl B=_{def} \pst A\mli B;\ \ \ A\bintimpl B =_{def} \bst A\mli B; \ \ \ A\intimpl B =_{def} \st A\mli B.\]
(Recurrences have the highest precedence, so $\pst A\mli B$ should be read as $(\pst A)\mli B$, and similarly for $\bst,\st$.) We refer to $\pst$ as {\bf parallel recurrence}, and refer to $\bst$ and $\st$ as {\bf branching recurrences}. Namely, $\bst$ can be called {\bf countable} branching recurrence, and $\st$ called {\bf uncountable} branching recurrence. $\st$ and $\pst$ have been defined in some earlier literature on computability logic (see, e.g., \cite{Japfin}). On the other hand, $\bst$, as a full-fledged citizen of computability logic, is first officially introduced in the present  paper (see also ``Historical remarks'' at the end of this section).  

Let us start with taking a closer (than done in Section 1) intuitive look at how $\st$ and $\pst$ compare.  
Imagine a computer that has a program successfully playing $\chess$. The resource that such a computer provides is obviously stronger than just $\chess$: a reasonable operating system would allow to simultaneously run as many parallel sessions of $\chess$ as the user needs, while  
$\chess$, as such, only assumes a single play.  This is what is captured by the parallel recurrence $\pst\chess$.  A more advanced operating system, however, would in addition also make it possible to branch/replicate each particular stage of each particular session, i.e. create any number of ``copies" of any already reached position
of the multiple parallel plays of $\chess$, thus giving the user 
the possibility to try different continuations from the same position. What corresponds to this intuition is the branching recurrence $\st\chess$. 

As noted earlier, the user of the resource $\st A$  does not have to restart $A$ from the very beginning every time it wants to reuse it; rather, it is (essentially) allowed to backtrack to any of the previous --- not necessarily starting --- positions and try a new continuation from there, thus depriving the adversary of the possibility to reconsider the moves it has already made in that position. This is in fact the type of reusage every purely software resource allows or would allow in the presence of an advanced operating system and unlimited memory:
one can start running process $A$; then fork it  
at any stage  thus creating two threads  that have a common past but possibly diverging futures  (with the possibility to treat one of the threads as 
a ``backup copy'' and preserve it for backtracking purposes); then further fork any of the branches at any time; and so on. The less flexible type of reusage of $A$ assumed by $\pst A$, on the other hand, is closer to what infinitely many autonomous 
physical resources would naturally offer, such as an unlimited number of independently acting robots each performing task $A$, or an unlimited number of computers with limited memories, each one only capable of and responsible for running a single thread 
of process $A$. Here  the effect of replicating/forking an advanced stage of $A$ cannot be achieved unless, by good luck, 
there are two identical copies of the stage, meaning that the corresponding two robots or computers have so far acted in precisely the same ways.

The difference between the countable ($\bst$) and uncountable ($\st$) versions of branching recurrence appears to be much more subtle than the difference between the parallel ($\pst$) and branching ($\st$ or $\bst$) sorts of recurrence. In fact, the above intuitive-level discussion of $\pst$ vs. $\st$ is just as valid for $\pst$ vs. $\bst$. Yet, $\st$ and $\bst$ turn out to induce dramatically different logics, even if those logics coincide 
when $\bintimpl$ or $\intimpl$ (or $\pintimpl$) is the only connective in the logical vocabulary.
The following is an example of a principle which could be shown to be valid with $\st$ but invalid with $\bst$ as well as with $\pst$:
\[\st\cost P\mli \cost\st P\]
($\cost$ abbreviates $\gneg\st\gneg$, where $\gneg$ is the ``role switch'' operation). And, as we started discussing differences between the principles validated by the different sorts of recurrences, here comes an example of a principle which can be shown to be valid with $\pst$ but invalid with either $\st$ or $\bst$: 
\[P\mlc \pst(P\mli Q\mlc P)\mli \pst Q\]
($\mlc$, called {\em parallel conjunction}, is a computability-logic counterpart of the tensor of linear logic. $A\mlc B$ means a parallel play of $A$ and $B$, where $\pp$ has to win in both plays to be the winner in the overall game). These are just isolated examples, and finding a systematic deductive characterization of all valid principles that involve recurrence operations remains a great challenge in computability logic.

Before we move to more examples illustrating differences between $\pst,\bst$ and $\st$,  it would be a good idea to first define the three operations under question.    

Formally, $\pst A$ is defined as the infinite conjunction $A\mlc A\mlc A\mlc\ldots$, where $A_0\mlc A_1\mlc A_2\mlc\ldots$ is a straightforward generalization of the just-mentioned parallel conjunction $\mlc$ operation from the binary case to the infinite case. 

Defining the branching recurrences takes more work.  In semiformal terms, a play of $\st A$  starts as an ordinary play of game $A$. At any time, however, player $\oo$ is allowed to make a ``replicative move'', which creates two copies of the current position $\Phi$ of $A$. From that point on, the game turns into two games played in parallel, each continuing 
from position $\Phi$. We use the bits $0$ and $1$ to denote those two threads, that --- using our earlier words --- have a common past (position $\Phi$) but possibly diverging futures. Again, at any time, $\oo$ can further branch either thread, creating two copies of the current position in that thread. If thread $0$ was branched, the resulting two threads will be denoted by $00$ and $01$; and if the branched thread was $1$, then the resulting threads will be denoted by $10$ and $11$. And so on: at any time, $\oo$ may split any of the existing threads $w$ into two threads $w0$ and $w1$. Each thread in the eventual run of the game will be thus denoted by a (possibly infinite) bit string. The game is considered won by $\pp$ if it wins $A$ in each of the threads; otherwise the winner is $\oo$. 

To each infinite bit string $w$ may thus correspond a separate run of $A$ in thread (represented by) $w$ and, as there are uncountably many infinite bit strings, uncountably many parallel runs of $A$ may be generated when playing $\st A$ up. Let us call a bit string $w$ {\bf essentially finite} if it contains only a finite number of $1$s; otherwise we say that $w$ is {\bf essentially infinite}. We extend these terms from bit strings to the corresponding threads in the play of $\st A$. The definition of $\st A$ thus requires from $\pp$ to win $A$ in all --- whether they be essentially finite or essentially infinite --- threads. All it takes to turn that definition into a definition of $\bst$ is to relax that requirement and, when determining the winner, only look at essentially finite threads. Since there are only countably many essentially finite bit strings, only countably many runs of $A$ are generated --- more precisely, only countably many runs of $A$ are of relevance --- in $\bst A$. This completes our semiformal definition/explanation of $\bst$.

In fully formal terms, both $\st A$ and $\bst A$ have the same structures ({\bf Lr} components). There are two types of legal moves in (legal) positions  of either game: (1) replicative and (2) non-replicative. To define these, let us agree that by an {\em active node}\footnote{Intuitively, an active node is (the name of) an already existing thread of a play over $A$.} of a position $\Phi$ we mean a bit string $w$ such that $w$ is either empty,\footnote{Intuitively, the empty string is the name/address of the initial thread; all other threads will be descendants of that thread.} or else is $u0$ or $u1$ for some bit string $u$ such that $\Phi$ contains the move $\col{u}$. A replicative move can only be made by (is only legal for) $\oo$, and such a move in a given position $\Phi$ should be $\col{w}$, where $w$ is an active node of $\Phi$ and $\Phi$ does not already contain the same move $\col{w}$.\footnote{The intuitive meaning of move $\col{w}$ is splitting thread $w$ into $w0$ and $w1$, thus ``activating'' these two new nodes/threads.} As for non-replicative moves, they can be made by either player. Such a move by a player $\xx$ in a given position $\Phi$ should be $w.\alpha$, where   $w$ is an active node of $\Phi$ and $\alpha$ is a move such that, for any infinite bit string $v$, $\alpha$ is a legal move by $\xx$ in position $\Phi^{\preceq wv}$ of $A$.\footnote{The intuitive meaning of such a move $w.\alpha$ is making move $\alpha$ in thread $w$ and all of its (current or future) descendants.} Here  and later, for a run $\Theta$ and bit string $x$, $\Theta^{\preceq x}$ means the result of deleting from $\Theta$ all moves except those that look like $u.\beta$ for some initial segment $u$ of $x$, and then further deleting the prefix ``$u.$'' from such moves.\footnote{Intuitively, $\Theta^{\preceq x}$ is the run of $A$ that has been played in thread $x$, if such a thread exists (has been  generated); otherwise, $\Theta^{\preceq x}$ is the run of $A$ that has been played in (the unique) existing thread which (whose name, that is) is
some initial segment of $x$.} As for the contents (the {\bf Wn} components) of these games, a legal run $\Gamma$ of $\st A$ is considered won by $\pp$ iff, for every infinite bit string $v$, $\Gamma^{\preceq v}$ is a $\pp$-won run of $A$.  And a legal run $\Gamma$ of $\bst A$ is considered won by $\pp$ iff, for every infinite but essentially finite bit string $v$, $\Gamma^{\preceq v}$ is a $\pp$-won run of $A$. This completes our definition of $\st$ and $\bst$.

As we just saw, the definition of $\bst$ is obtained from the definition of $\st$ by merely inserting the words ``but essentially finite''. But, again, trying to analyze the rather technical definition given in the above paragraph may not be a good idea for a reader of this paper. Relying, instead, on the informal explanations that we provided should be sufficient.

To see the distance between $\st$ and $\bst$, following Vereshchagin \cite{Ver}, let us consider any set $S$ of natural numbers (identified with their decimal representations), such that 
$S$ is not recursively enumerable. Let $A$ be the game where 
only $\pp$ has legal moves, each legal move being a(ny) natural number. A given run of this game is considered won by $\pp$ iff the set of the moves it makes in it equals 
$S$. In other words, $\pp$ wins iff it enumerates $S$. Now let us look at the games $\cost A$ and $\cost^{\aleph_0} A$, where $\cost=\gneg \st\gneg$ and $\cost^{\aleph_0}=\gneg\st^{\aleph_0}\gneg$. That is, $\cost A$ is the same as $\st A$, only here it is $\pp$ rather than $\oo$ who can create new threads (make replicative moves), 
and whose adversary needs to win $A$ in each of the threads to be the winner in the overall game. Similarly for $\cost^{\aleph_0}$.  Obviously $\pp$ has an effective winning strategy for $\cost A$, consisting in enumerating all of the 
$2^{\aleph_0}$ sets of natural numbers, one per each of the $2^{\aleph_0}$ threads that it may create in $\cost A$. On the other hand, $\pp$ does not have an effective winning strategy for $\cost^{\aleph_0} A$. Otherwise, one would be able to recursively enumerate $S$ by selecting the (essentially finite) bit string $w$ representing a winning thread, and then listing the moves made in that thread. 

Our discussion would not be complete without also seeing a specific example illustrating the distance between the parallel and branching versions of recurrence.  The game $\pcost A$, which is a dual of $\pst A$ in the same sense as $\cost A$ is a dual of $\st A$, is defined as $A\mld A\mld A\mld\ldots$. This can be thought of as a parallel play of game $A$ on infinitely many boards: \#0, \#1, \#2, \ldots. 
$\oo$ wins it iff it wins $A$ on each of the boards.  Where $f(x)$ is a total function from natural numbers to natural numbers, $\ada x\ade y \bigl(y=f(x)\bigr)$ denotes a game every legal run of which consists of (at most) two moves.\footnote{In computability logic, $\ada$ is called {\em choice universal quantifier}, and $\ade$ called {\em choice existential quantifier}. The smaller versions $\adc$ and $\add$ of the same symbols stand for what are called {\em choice conjunction} and {\em choice disjunction}, respectively. The choice operators of computability logic are reminiscent of the additive operators of linear logic.} The first move is by $\oo$, and the move is an arbitrary number $m$. The second move is by $\pp$, who should name a number $n$. $\pp$ wins iff $n$ equals $f(m)$. $\pp$ has an effective winning strategy that works for both $\cost\ada x\ade y\bigl(y=f(x)\bigr)$ and $\cost^{\aleph_0}\ada x\ade y\bigl(y=f(x)\bigr)$. It consists in waiting till the adversary makes a move $m$, after which $\pp$ creates infinitely (but countably) many threads, and tries all possible responses --- all possible values for $n$, that is --- in those threads, one response per thread. Similarly, where $B(x)$ is a predicate, $\ada x(\gneg B(x)\add B(x))$ denotes a game where the first move (again), consisting of choosing a number $m$, is by $\oo$. The second move is by $\pp$, who should choose between $0$ and $1$. $\pp$ wins iff $B(m)$ is false and $0$ was chosen, or $B(m)$ is true and $1$ was chosen. $\pp$'s effective winning strategy for both $\cost \ada x(\gneg B(x)\add B(x))$ and $\cost^{\aleph_0} \ada x(\gneg B(x)\add B(x))$ is that it waits till  the adversary makes a move $m$, after which $\pp$ creates two threads, making move $0$ in one thread and move $1$ in the other thread. The same trick, however, fails with $\pcost \ada x(\gneg B(x)\add B(x))$. For example, it fails when $\oo$ chooses 
$m=0$ on board \#0, $m=1$ on board \#1, $m=2$ on board \#2, etc.  Let us call this strategy of $\oo$ the {\em diversifying strategy}. Now, for any effective strategy $\cal M$ of $\pp$, using diagonalization, we can construct a particular predicate $B(x)$ such that $\pp$ loses $\pcost \ada x(\gneg B(x)\add B(x))$
against the diversifying strategy. Namely, we can define $B(i)$ (any $i$) to be true if $\cal M$ makes the move $0$ on board \#$i$ when playing against the diversifying strategy, and false otherwise. This guarantees that $\cal M$'s all responses to the adversary's moves are ``wrong''. A similar idea could be employed in showing that, for an appropriately selected $f$, the problem $\pcost\ada x\ade y\bigl(y=f(x)\bigr)$ has no algorithmic solution.\vspace{5pt}

{\bf Historical remarks.} Blass \cite{Bla72} was apparently the first to consider an operator in the style of the exponential operator $!$ of linear logic. He called it the {\em repetition operator} $R$. The game-semantical context in which $R$ was introduced was limited compared with the context that computability logic operates in. The main contextual difference is that Blass's games are {\em strict}, meaning games where in each position only one player may have (legal) moves. Computability logic, on the other hand, deals with the already mentioned more general type of {\em static games}. As opposed to strict games, static games are {\em free}, in the sense that generally both players may have legal moves in a given position. Furthermore, the recurrence operations 
(as well as the non-recurrence parallel operations $\mlc,\mld,\mla,\mle$) of computability logic generate properly free games even when applied to strict 
games, while Blass's operations, of course, preserve the strict property of games. However, if we disregard this difference and try to bring Blass's games and 
static games to some reasonable common denominator, Blass's repetition operation $R$ would apparently translate (whatever ``translate'' should precisely mean here) into $\bst$. The reason why it would not translate into $\pst$ is that $R$ is 
a branching operation in the proper sense, allowing effects such as backtracking. And the reason why $R$ would be less than an adequate counterpart of 
$\st$ is that $\st A$ allows $\bot$ to try a continuum of different runs of $A$, while that quantity is automatically limited to $\aleph_0$ in $R A$ (and artificially limited to $\aleph_0$ in $\bst A$, as we saw from the definition).

\section{Implicative intuitionistic logic} 
Where $\supset$ is one of the operators $\pintimpl$, $\bintimpl$ or $\intimpl$, a Gentzen-style axiomatization of the corresponding {\em implicative} (fragment of) {\em intuitionistic logic}, denoted by  ${\bf Int}^\supset$,  is {\bf CL7} --- only with $\supset$-sequents instead of $\mli$-sequents, of course --- plus the following single additional rule

\begin{center}
\begin{picture}(55,70)
\put(0,30){$\Gamma,\hspace{2pt}E,\hspace{2pt}E\ \Rightarrow \ F$}
\put(2,50){\bf Contraction}
\put(0,22){\line(1,0){63}}
\put(7,8){$\Gamma,\hspace{2pt}E\ \Rightarrow \ F$}
\end{picture}
\end{center}

Alternatively, ${\bf Int}^\supset$ could be chosen to be formulated exactly as {\bf CL7}, with the only difference that the antecedents of sequents in ${\bf Int}^\supset$ are seen as sets rather than multisets of formulas, which eliminates the need for explicitly stating contraction as an inference rule. 

\begin{theorem}\label{th1}
Let $\supset$ be any one of the operators $\pintimpl$, $\bintimpl$ or $\intimpl$. 
For any $\supset$-formula $H$, the following conditions are equivalent:

(i) \ ${\bf Int}^{\supset}\vdash H$.

(ii) $H$ is valid in computability logic, whether it be in the ordinary sense of validity or in the sense of 
uniform validity.   
\end{theorem}

\begin{proof} For $\hintp$ and $\hint$, this theorem was officially established in \cite{Japjsl}. As for $\hintb$, as it happens, the proof of the soundness and completeness of $\hint$ given in \cite{Japjsl}, in fact, is also a proof of the soundness and completeness of $\hintb$: a simple re-reading of that proof reveals that virtually no step in it relies on the fact that we deal with the uncountable rather than the countable version of reduction. \end{proof} 

{\bf Historical remarks and further discussions.} The above-mentioned result of \cite{Japjsl} for $\hint$ was further strengthened in \cite{Propint}, where  soundness and completeness (with respect to the semantics of computability logic) was proven for the full propositional fragment of intuitionistic logic.
With only one or two months' delay, Vereshchagin \cite{Ver} came up with an alternative and shorter proof of the same result. It should be noted, however, that in his work Vereshchagin  modified the ``canonical'' definitions of computability logic quite a bit, which essentially resulted in interpreting intuitionistic implication as $\bintimpl$ rather than $\intimpl$. Moreover, in an attempt to simplify things, Vereshchagin further limited games to strict ones, essentially defining the 
$\bst$ component of $\bintimpl$ as something closer to Blass's repetition operator $R$ than to $\bst$ in our present precise sense.
In view of (and despite) the above-said, Vereshchagin's proof, with certain technical adjustments, can be considered an alternative proof of the $\hintb$ part of Theorem \ref{th1}.  

The soundness proof for $\hint$ found in \cite{Japjsl}, in fact, can be dramatically simplified when we are concerned with 
$\bintimpl$ rather than $\intimpl$. Specifically, a lemma on which the soundness proof given in \cite{Japjsl} (as well as similar proofs given in \cite{Bla72} and \cite{Ver}) relies is about the validity of the principle $\st A\mli \st\st A$. And a strict proof of that lemma, presented in \cite{Japfin}, takes several pages.  On the other hand, a proof of the validity of the same principle 
for $\bst$ instead of $\st$ would not take more than just a paragraph, as it did (for Blass's or Vereshchagin's versions of $\bst$) in \cite{Bla72} or \cite{Ver}.   

The above-said also applies to the proof of the soundness and completeness of the full propositional intuitionistic logic given in 
\cite{Propint}. That proof officially is for the case when the intuitionistic implication is read as $\intimpl$. However, the same proof is just as good 
for $\bintimpl$ as well. Similarly, Vereshchagin's \cite{Ver} completeness proof can be adapted to either interpretation $\intimpl$, $\bintimpl$ of intuitionistic implication. The same cannot be said about Vereshchagin's soundness proof though: as noted above, proving soundness when intuitionistic implication is read as $\intimpl$ rather than $\bintimpl$ takes considerably greater efforts.  

Finally, for reasons similar to the above, the soundness proof for the full first-order intuitionistic calculus given in \cite{int1}, is equally good for either reading $\intimpl,\bintimpl$ of intuitionistic implication.

\section{Conclusion} Computability logic is a semantically conceived approach with the ambition to be ``a formal theory of computability in the same sense as classical logic is a formal theory of truth'' (\cite{Japtocl1}).  As such, it does not yet have a sufficiently developed syntax, and among the main current  objectives of computability logic as a research program is to find axiomatizations for various natural fragments of the set of formulas validated by its semantics. The language of computability logic, with logical operators standing for operations on computational problems, is very rich and, in fact, open-ended. Identifying the most natural and potentially useful new operators to be included in it is another direction on which efforts within the project continue to be focused. 

The present paper contributes to both of the above directions. Within the second direction, it officially introduces the Blass-style (\cite{Bla72,Bla92}) {\em countable recurrence} operator $\bst$ and the associated reduction operator $\bintimpl$. It also outlines the idea of finite and bounded versions of recurrence operators together 
with the associated reduction operators. The three  other main reduction  operators $\mli,\pintimpl$ and $\intimpl$ studied in the paper have been introduced earlier. The variety of reduction operators captures various flavors of our intuition of algorithmically reducing one computational problem to another, with $\mli$ being the strongest form of reduction and $\intimpl$ being the weakest form.  

Within the first direction, this paper establishes two results. According to one result, a greater part of which was established earlier, the logic induced by each of $\pintimpl,\bintimpl,\intimpl$ is exactly the implicative fragment of Heyting's intuitionistic calculus. And, according to the other theorem, the logic induced by $\mli$ is the same calculus but with the contraction rule removed. The philosophical summary of these results is that, despite the significant semantical varieties within the basic group of reduction operations, they (when taken in isolation) only generate two kinds of logical behavior, depending on whether resource/antecedent reusage is allowed ($\pintimpl,\bintimpl,\intimpl$) or not ($\mli$). Syntactically, those two behaviors are precisely accounted for 
by the mere presence or absence of contraction in Gentzen-style axiomatizations.

\end{document}